\documentclass[12pt,a4paper]{article}
\usepackage[latin1]{inputenc}
\usepackage[english]{babel}
\usepackage{amsmath}
\usepackage{amsfonts}
\usepackage{amssymb}
\usepackage{indentfirst}
\usepackage{dsfont}
\newtheorem{theorem}{Theorem}[section]
\newtheorem{lemma}[theorem]{Lemma}
\newtheorem{proposition}[theorem]{Proposition}
\newtheorem{corollary}[theorem]{Corollary}
\newenvironment{proof}[1][Proof]{\begin{trivlist}
\item[\hskip \labelsep {\bfseries #1}]}{\end{trivlist}}
\newenvironment{definition}[1][Definition]{\begin{trivlist}
\item[\hskip \labelsep {\bfseries #1}]}{\end{trivlist}}

\newenvironment{remark}[1][Remark]{\begin{trivlist}
\item[\hskip \labelsep {\bfseries #1}]}{\end{trivlist}}
\newcommand{\lin}{\operatorname{lin}}
\newcommand{\qed}{\nobreak \ifvmode \relax \else
      \ifdim\lastskip<1.5em \hskip-\lastskip
      \hskip1.5em plus0em minus0.5em \fi \nobreak
      \vrule height0.75em width0.5em depth0.25em\fi}
\author{Fabrice Gamboa, Jean-Michel Loubes and Paul Rochet}
\title{Maximum Entropy Estimation for Survey sampling}
\date{}
\newcommand{\mykeywords}{
Survey Sampling, Inverse Problems, Maximum Entropy Method}

\newcommand{\mysubjclass}{
 62F12,
 62D05,
 94A17 
 }
\begin{document}
\maketitle
\begin{abstract}
Calibration methods have been widely studied in survey sampling over the last decades. Viewing calibration as an inverse problem, we extend the calibration technique by using a maximum entropy method. Finding the optimal weights is achieved by considering random weights and looking for a discrete distribution which maximizes an entropy under the calibration constraint. This method points a new frame for the computation of such estimates and the investigation of its statistical properties.
\end{abstract}
$ $ \vskip .1in
\noindent \textbf{Keywords}: \mykeywords\\
\noindent \textbf{Subject Class. MSC-2000} : \mysubjclass

\section*{Introduction}
Calibration is a well spread method to improve estimation in survey sampling, using extra information from an auxiliary variable. This method provides approximately unbiased estimators with variance smaller than that of the usual Horvitz-Thompson estimator (see for example \cite{MR0413318}).  Calibration has been introduced  by Deville and S{\"a}rndal in \cite{MR1173804}, extending an idea of \cite{Deville5}. For general references, we refer to \cite{MR2274771}, \cite{MR2024768} and  for an extension to variance estimation to \cite{singh17}.

\noindent
Finding the solution to a calibration equation involves minimizing an energy under some constraint. 
More precisely, let $s$ be a random sample of size $n$ drawn from a population $U$ of size $N$, $y$ is the variable of interest and  $x$ is a given auxiliary variable, for which the mean $t_x$ over the population is known. Further, let $d \in \mathbb R^n$ be the standard sampling weights (that is the  Horvitz-Thompson ones). Calibration derives an estimator $\hat t_y = N^{-1}  \sum_{i\in s}w_i y_i$ of the population mean $t_y$ of $y$. The weights $w_i$ are chosen to minimize a dissimilarity (or distance) $\mathcal D(.,d)$ on $\mathbb R^n$ with respect to the Horvitz-Thompson weights $d_i$ and under the constraint 
\begin{equation}
N^{-1}  \sum_{i\in s}w_i x_i = t_x.
\label{cadix}
\end{equation}

Following \cite{theberge18}, we will view here calibration  as a linear inverse problem. In this paper, we use Maximum Entropy Method on the Mean (MEM) to build the calibration weights. Indeed, MEM is a strong machinery  for solving linear inverse problems. It tackles a linear inverse problem by finding a measure maximizing an entropy under some suitable constraint. 
It has been extensively studied and used in many applications, see for example  \cite{MR1397223}, \cite{MR1331146}, \cite{MR1429928}, \cite{Kitamura14}, \cite{Gamboa9}, \cite{Fermin8} or \cite{MR2115273}.

\noindent 
Let us roughly explain how MEM works in our context. First we fix a {\it prior} probability measure $\nu$ on $\mathbb{R}^n$ with mean value equal to $d$. Then, the idea is to modify the weights in the sample mean in order to get a representative sample for the auxiliary variable $x$, but still being as close as possible to $d$, which have the desirable property of yielding an unbiased estimate for the mean. So, we will look for a {\it posterior} probability measure minimizing the entropy (or Kullback information) with respect to  $\nu$  and satisfying a constraint related to \eqref{cadix}. It appears that the MEM estimator is in fact a specific calibration estimator for which the corresponding dissimilarity $\mathcal D(.,d)$ is determined by the choice of the prior distribution $\nu$.
Hence, the MEM methodology provides a general Bayesian frame to fully understand calibration procedures in survey sampling where  the different choices of dissimilarities appear as different choices of prior distributions.

\noindent
An important problem when studying calibration methods is to understand the amount of information contained in the auxiliary variable. 
Indeed,  it appears that the relationships between the variable to be estimated and the auxiliary variable are crucial to improve estimation (see for example  \cite{MR2236453} or \cite{MR2274771}). When complete auxiliary information is available, increasing the correlation between the variables is made possible by replacing the auxiliary variable $x$ by some function of it, say $u(x)$. So, we consider efficiency issues for a collection of calibration estimators, depending on both the choice of the auxiliary variable and the dissimilarity. Finally, we provide an optimal way of building an efficient estimator using the MEM methodology. \\

The article falls into the following parts. The first section recalls the calibration method in survey sampling, while the second exposes the MEM methodology in a general framework, and its application to calibration and instrument estimation. Section~\ref{sec3} is devoted to the choice of a data driven calibration constraint in order to build an efficient calibration estimator. It is shown to be optimal under strong asymptotic assumptions on the sampling design. 
Simulations illustrate previous results in Section~\ref{sec4} while the proofs are postponed to Section~\ref{sec5}.

\section{Calibration Estimation of a linear parameter}\label{sec1}
Consider a large population $U = \left\lbrace 1,...,N \right\rbrace $ and an unknown characteristic $y = (y_1,...,y_N) \subset \mathbb R^N$. Our aim is to estimate its mean $t_y := N^{-1} \sum_{i \in U} y_i $ when only a random subsample $s$ of the whole population is available.  So the observed data are $(y_i)_{i \in s}$. The sampling design is the probability distribution $p$ defined for each subset $s \subset U$ as the probability $p(s)$ that $s$ is observed. We assume that $\pi_i := p( i \in s ) = \sum_{s, \ i \in s } p(s)$ is strictly positive for all $i \in U$, so $ d_i = 1/\pi_i $ is well defined. A standard estimator of $t_y$ is given by the Horvitz-Thompson estimator:
$$ \hat{t}_y^{HT} = N^{-1} \sum_{i \in s} \frac{y_i}{\pi_i}  = N^{-1} \sum_{i \in s} d_i y_i .$$
This estimator is unbiased and is widely used for practical cases, see for instance \cite{Deville5} for a complete survey.\\

 Suppose that it exists an auxiliary vector variable $x = (x_1,...,x_N)$, that is entirely observed and set $t_x = N^{-1} \sum_{i \in U} x_i \in \mathbb{R}^k$. If the Horvitz-Thompson estimator of $t_x$, $ \hat t_x^{HT} = N^{-1} \sum_{i \in s} d_i x_i $ is far from the true value $t_x$, it may imply that the sample does not describe well the behavior of the variable of interest in the total population. So, to prevent biased estimation due to bad sample selection, inference on the sample can be achieved by considering a  modification of the weights of the individuals chosen in the sample.\\
\indent One of the main methodology used to correct this effect is the calibration method, (see \cite{MR1173804}). The {\it bad sample effect} is corrected by deriving new weights for the sample mean, but still being close to the $d_i$'s to get a small bias. For this, consider a class of weighted estimators $ N^{-1} \sum_{i \in U} w_i y_i$ where the weights $w = (w_i)_{i \in s}$ are selected to be close to $d=(d_i)_{i \in s}$ under the calibration constraint
$$  N^{-1} \sum_{i \in s} w_i x_i = t_x.    $$
There are two basic components in the construction of calibration estimators, namely a dissimilarity and a set of calibration equations. Let $w \mapsto \mathcal D(w,d)$ be a dissimilarity between some weights and the Horvitz-Thompson ones. Assume that this dissimilarity is minimal for $w_i=d_i$. The method consists in choosing weights minimizing $\mathcal D(.,d)$ under the constraint $N^{-1} \sum_{i \in s} w_i x_i = t_x$. \\

A typical dissimilarity is the $\chi^2$ distance $w \mapsto \sum_{i \in s} (\pi_i w_i - 1)^2/(q_i \pi_i)$ for $(q_i)_{i \in s}$ a positive smoothing sequence (see \cite{MR1173804}). So the new estimator is defined as $\hat t_y=N^{-1} \sum_{i \in s} \hat w_i y_i$, where the weights $\hat w_i$ minimizes
$ \mathcal D(w,d) = \sum_{i \in s} (\pi_i w_i - 1)^2/q_i \pi_i$
under the constraint $N^{-1} \sum_{i \in s} \hat w_i x_i = t_x$. Denote by $a^t$ the transpose of $a$, the solution of this minimization problem is given by
$$  \hat t_y = \hat t_y^{HT} + (t_x - \hat t_x^{HT})^t \hat B,    $$
where $ \hat B = \textstyle \left[ \sum_{i \in s} q_i d_i  x_i x_i^t\right]^{-1} \sum_{i \in s} q_i d_i  y_i x_i$. Note that this is a generalized regression estimator. It is natural to consider alternative measures, which are given in \cite{MR1173804}. We first point out that the existence of a solution to the constrained minimization problem depends on the choice of the dissimilarities. Then, different choices can lead to weights with different behaviors, different ranges of values for the weights that may be found unacceptable by the users. We propose an approach where dissimilarities have a probabilistic interpretation. This highlights the properties of the resulting estimators.

\section{Maximum Entropy for Survey Sampling}\label{sec2}
\subsection{MEM methodology}\label{s:MEMintro}
Consider the problem of recovering an unknown measure $\mu$ on a measurable space $\mathcal{X}$ under moment conditions.  We observe a random sample $T_1,...,T_n \sim \mu$. For a given function $\mathrm x: \mathcal X \rightarrow \mathbb R^k$ and a known quantity $t_x \in \mathbb R^k$, we aim to estimate $\mu$ satisfying
\begin{equation}
\label{chile1} \int_\mathcal X \mathrm x(t) d \mu(t) = t_x. \end{equation}
This issue belongs to the class of generalized moment problems with convex constraints (we refer to \cite{MR1408680} for general references), which can be solved using maximum entropy on the mean (MEM). The general idea is to modify the empirical distribution $\mu_n=n^{-1}\sum_{i=1}^n \delta_{T_i}$ in order to take into account the additional information on $\mu$ given by the moment equation \eqref{chile1}. For this, consider weighted versions of the empirical measure $n^{-1} \sum_{i=1}^n p_i \delta_{T_i}$  for weights $p_i$ properly chosen. The MEM estimator $\hat \mu_n$ of $\mu$ is a weighted version of $\mu_n$, where the weights are the expectation of a random variable $P = (P_1,...,P_n)$, drawn from a finite measure $\nu^*$ close to a \textit{prior} $\nu$. This prior distribution conveys the information that $\hat \mu_n$ must be close to the empirical distribution $\mu_n$.
More precisely, let first define the relative entropy or Kullback information between two finite measures $Q,R$ on a space $(\Omega,\mathcal{A})$ by setting
$$ K(Q,R) = \begin{cases} \int_\Omega  \log  \left( \frac{d Q}{dR} \right)  dQ - Q(\Omega) + 1 \: &  {\rm if } \ Q \ll R \\
 +\infty  \: &  {\rm otherwise.}  \end{cases} $$ 
Since this quantity is not symmetric, we will call it the relative entropy of $Q$ with respect to $R$. Note also that, among the literature in optimization, the relative entropy is often defined as the opposite of the entropy defined above, which explains the name of maximum entropy method, while with our notations, we consider the minimum of the entropy. \\

\noindent Given our prior $\nu$, we now define $\nu^*$ as the measure minimizing $K(.,\nu)$ under the constraint that the linear constraint holds in mean:
 $$ \mathbb E_{\nu^*} \left[ \textstyle n^{-1} \sum_{i =1}^n P_i \mathrm x_i \right] = \displaystyle \dfrac{1}{\nu^*(\mathbb R^n)} \int_{\mathbb R^n} \textstyle \left[ n^{-1} \sum_{i =1}^n p_i \mathrm
 x_i \right] d \nu^*(p_1,...,p_n)= t_x,  $$
where we set $\mathrm x_i = \mathrm x(T_i)$. We then build the MEM estimator $\hat \mu_n = n^{-1} \sum_{i = 1}^n \hat p_i \delta_{T_i} $, where $\hat p = (\hat p_1,...,\hat p_n) = \mathbb E_{\nu^*} (P)$.\\

This method provides an efficient way to estimate some linear parameter $t_y  = \int_\mathcal X \mathrm y d \mu$ for $\mathrm y: \mathcal X \rightarrow \mathbb R$ a given map. The empirical mean $\overline{\mathrm y} = \int_\mathcal X \mathrm y d\mu_n$ is an unbiased and consistent estimator of $t_y$ but may not have the smallest variance in this model. We can improve the estimation by considering the MEM estimator $ \hat t_y = n^{-1} \sum_{i =1}^n \hat p_i \mathrm y_i$, which has a lower variance than the empirical mean and is asymptotically unbiased (see \cite{MR1429928}). \\

In many actual situations, the function $\mathrm x$ is unknown and only an approximation to it, say $\mathrm x_m$, is available. Under regularity conditions, the efficiency properties of the MEM estimator built with the approximate constraint have been studied in \cite{devore2} and \cite{l2}, introducing the approximate maximum entropy on the mean method (AMEM). More precisely, the AMEM estimate of the weights is defined as the expectation of the variable $P$ under the distribution $\nu^*_m$ minimizing $K(.,\nu)$ under the approximate constraint 
\begin{equation}
\label{chile2} \mathbb E_{\nu^*_m} \left[ \textstyle n^{-1} \sum_{i =1}^n P_i \ \! \mathrm x_m(T_i) \right] =  t_x. \end{equation}
It is shown that, under assumptions on $ \mathrm x_m$, the AMEM estimator of $t_y$ obtained in this way is consistent as $n$ and $m$ tends to infinity. This procedure enables to increase the efficiency of a calibration estimator while remaining in a Bayesian framework, as shown in Section \ref{sec32}.

\subsection{Maximum entropy method for calibration}\label{sec22}

Recall that our original problem is to estimate the population mean $t_y = N^{-1} \sum_{i \in U} y_i$ based on the observations $\left\{y_i, i \! \in \! s \right\}$ and auxiliary information $\left\{x_i, i \! \in \! U \right\}$. We introduce the following notations:
$$ \mathrm y_i =  n N^{-1} d_i y_i, \ \mathrm x_i =  n N^{-1} d_i x_i, \ p_i = \pi_i w_i. $$
Note that the variables of interest are rescaled to match the MEM framework. The weights $(p_i)_{i \in s}$ are now identified with a discrete measure on the sample $s$. The Horvitz-Thompson estimator $\hat t_y^{HT} = N^{-1} \sum_{i \in s} d_i y_i = n^{-1} \sum_{i \in s} \mathrm y_i$ is the preliminary estimator we aim at improving. The calibration constraint $n^{-1} \sum_{i \in s} p_i \mathrm x_i = t_x$ stands for the linear condition satisfied by the discrete measure $(p_i)_{i \in s}$. So, it appears that the calibration problem follows the pattern of maximum entropy on the mean. Let $\nu$ be a prior distribution on the vector of the weights $(p_i)_{i \in s}$. The solution $\hat p = (\hat p_i)_{i \in s}$ is the expectation of the random vector $P=(\pi_i W_i)_{i \in s}$ drawn from a \textit{posterior} distribution $\nu^*$, defined as the minimizer of the Kullback information $K(.,\nu)$ under the condition that the calibration constraint holds in mean
$$ \mathbb E_{\nu^*} \left[ \textstyle n^{-1} \sum_{i \in s} P_i \mathrm x_i \right]= \mathbb E_{\nu^*} \left[ \textstyle N^{-1} \sum_{i \in s} W_i x_i \right] = t_x.     $$
We take the solution $\hat p  = \mathbb E_{\nu^*} (P)$ and define the corresponding MEM estimator $\hat t_y$ as
$$ \hat t_y =  n^{-1} \sum_{i \in s} \hat p_i \mathrm y_i =  N^{-1} \sum_{i \in s} \hat w_i y_i,$$ 
where we set $\hat w_i = d_i \hat p_i$ for all $i \! \in \! s$. Under the following assumptions, we will show in Theorem \ref{prop1} that maximum entropy on the mean gives a Bayesian interpretation of calibration methods.\\
 
The random weights $P_i,i \in s$ (and therefore the $W_i,i \in s$) are taken independent and we denote by $\nu_i$ the prior distribution of $P_i$. It follows that $\nu = \otimes_{i \in s} \nu_i $. Moreover, all prior distributions $\nu_i$ are integrable with mean $1$. This last assumption conveys that $\hat p_i$ must be close to $1$, equivalently, $\hat w_i = d_i \hat p_i$ must be close to the Horvitz-Thompson weight $d_i$. \\ 
Let $\varphi: \mathbb R \rightarrow \mathbb R$ be a closed convex map, the convex conjugate $\varphi^*$ of $\varphi$ is defined as
$$ \forall s \in \mathbb R, \  \varphi^*(s) = \underset{t \in \mathbb{R}}{\text{sup}} ( st - \varphi(t)). $$
For $\nu$ a probability measure on $\mathbb{R}$, we denote by $\Lambda_\nu$ the log-Laplace transform of $\nu$:
$$ \Lambda_\nu(s) = \log \int e^{s x}  d\nu(x), \ s \in \mathbb R. $$ 
Its convex conjugate $\Lambda^*_\nu$ is the Cramer transform of $\nu$. Moreover, denote by $S_\nu$ the interior of the convex hull of the support of $\nu$ and let $ D(\nu)  =  \left\lbrace s \in \mathbb R: \ \Lambda_\nu(s) <  \infty \right\rbrace$. In the sequel, we will always assume that $\Lambda_{\nu_i}$ is essentially smooth (see \cite{rock}) for all $i$, strictly convex and that $\nu_i$ is not concentrated on a single point. The last assumption means that if $D(\nu_i) = (- \infty;\alpha_i)$, ($\alpha_i \leq + \infty$), then $\Lambda_{\nu_i}'(s)$ goes to $+ \infty$ whenever $\alpha_i < + \infty$ and $s$ goes to $\alpha_i$. Notice that, under these assumptions, $\Lambda_{\nu_i}'$ is an increasing bijection between the interior of $D(\nu_i)$ and $S_{\nu_i}$. Moreover, we have the functional equalities $ ({\Lambda_{\nu_i}^*}')^{-1} = \Lambda_{\nu_i}' $ and $({\Lambda_{\nu_i}^*})^* = \Lambda_{\nu_i}$. 

\begin{definition}\textbf{:} We say that the optimization problem is feasible if there exists a vector $\delta = (\delta_i)_{i \in s} \in \otimes_{i \in s} S_{\nu_i}$ such that:
$$  N^{-1} \sum_{i \in s} \delta_i x_i = t_x.  $$
\end{definition}
Under the last assumptions, the following proposition claims that the solutions $(\hat w_i)_{i \in s}$ are easily tractable.

\begin{theorem}[survey sampling as MEM procedure]\label{prop1} Assume that the optimization problem is feasible. The MEM estimator $\hat w=(\hat w_1,...,\hat w_n)$ minimizes over $\mathbb R^n$
$$ (w_1,...,w_n) \mapsto \sum_{i \in s} \Lambda^*_{\nu_i}  (\pi_i w_i)  $$
under the constraint $N^{-1} \sum_{i \in s} \hat w_i x_i = t_x$. 
\end{theorem}

Hence, we point out that maximum entropy on the mean method leads to calibration estimation, where the dissimilarity is determined by the Cramer transforms $\Lambda^*_{\nu_i}, i \in s $ of the prior distributions $\nu_i$.

\begin{remark}\textbf{: (relationship with Bregman divergences)} Taking the priors $\nu_i$ in a certain class of measures may lead to specific dissimilarities known as Bregman divergences. We refer to \cite{MR2115273} for a definition. In the MEM method, there are two different kinds of priors for which the resulting dissimilarity may be seen as a Bregman divergence. Let $\nu$ be a probability measure with mean $1$ and such that $\Lambda_\nu$ is a strictly convex function. Then, $\Lambda_\nu^*$ enables to define a Bregman divergence. It will play the role of the dissimilarity resulting from the MEM procedure in the two following situations.\\
First, consider priors $\nu_i, i \in s$ all taken equal to $\nu$. It is a simple calculation to see that the assumptions made on $\nu$ imply that $\Lambda_\nu^* (1) = {\Lambda_\nu^*} ' (1) = 0$. The resulting dissimilarity can thus be written as
$$ \mathcal D(w,d) = \sum_{i \in s} \Lambda_{\nu}^* (\pi_i w_i) = \sum_{i \in s} \left[ \Lambda_{\nu}^*  (\pi_i w_i) - \Lambda_{\nu}^{*} (1) - {\Lambda_{\nu}^*}'(1) ( \pi_i w_i - 1) \right].    $$
Here, we recognize the expression of the Bregman divergence between the weights $ \left\{ \pi_i w_i, \ i \in s \right\}$ and $1$ associated to the convex function $\Lambda^{*}_{\nu}$. \\
Another possibility is to take prior distributions $\nu_i$ lying in some suitable exponential family. More precisely, define the prior distributions as
$$\forall i \in s, \forall x \in \mathcal X, d \nu_i(x) = \exp (\alpha_i x + \beta_i) d \nu( d_i x),$$
where $\beta_i = - \Lambda_{\nu} ( {\Lambda_{\nu}^*} '(d_i))$ and $\alpha_i = d_i {\Lambda_{\nu}^*} '(d_i)$ are properly chosen so that $\nu_i$ is a probability measure with mean $1$. Here we recover after some computation the following dissimilarity
$$ \mathcal D(w,d) = \sum_{i \in s} \left[ \Lambda_{\nu}^*  (w_i) - \Lambda_{\nu}^{*} (d_i) - {\Lambda_{\nu}^*}'(d_i) ( w_i - d_i) \right],      $$
which is the Bregman divergence between $w$ and $d$ associated to $\Lambda^{*}_{\nu}$.
\end{remark}

\subsection{Bayesian interpretation of calibration using MEM}\label{sec22bis} 
In a classical presentation, calibration methods heavily rely on a distance choice. Here, this choice corresponds to different prior measures $(\nu_i)_{i \in s}$. We now see the probabilistic interpretation of some commonly used distances. \\  

\noindent \textbf{Stochastic interpretation of some usual calibrated survey sampling estimators}
\begin{enumerate}
\item Generalized Gaussian prior.\\
For a given positive sequence $q_i, i \! \in \! s$, let $W_i$ having a Gaussian distributions $\mathcal N(d_i,d_i q_i)$ which corresponds to $\nu_i \sim \mathcal N(1,\pi_i q_i)$. We get
\begin{eqnarray*}
\forall t \in \mathbb R, \ \Lambda_{\nu_i}(t) = \dfrac{q_i \pi_i t^2}{2} + t \ ; \ \Lambda_{\nu_i}^*(t) = \dfrac{(t-1)^2}{2 \pi_i q_i}
\end{eqnarray*}
The calibrated weights in that cases minimize the criterion
$$ \mathcal D_1(w,d) = \sum_{i \in s} \dfrac{(\pi_i w_i -1)^2}{q_i \pi_i} .$$
So, we recover the $\chi^2$ distance discussed in Section \ref{sec1}. This is one of the main distance used in survey sampling. The choice of the $q_i$ can be seen as the choice of the variance of the Gaussian prior. The larger the variance, the less stress is laid on the distance between the weights and the original Horvitz-Thompson weights.
\item Exponential prior.\\
We take a unique prior $\nu$ with an exponential distribution with parameter $1$. That is, $\nu = \nu^{\otimes n}$. We have in that case
$$ \forall t \in \mathbb R_+^*, \ \Lambda_{\nu}^*(t) = - \log t + t -1.     $$
This corresponds to the following dissimilarity
$$ \mathcal D_2(w,d) = \sum_{i \in s} - \log (\pi_i w_i) + \pi_i w_i.       $$
We here recognize the Bregman divergence between $ (\pi_i w_i)_{i \in s}$ and $1$ associated to $\Lambda_{\nu}^*$, as explained in the previous remark. A direct calculation shows that this is also the Bregman divergence between $w$ and $d$ associated to $\Lambda_{\nu}^*$. The two distances are the same in that case.
\item Poisson prior.\\
If we choose for prior $\nu_i=\nu, \forall i \in s$, where $\nu$ is the Poisson distribution with parameter $1$, then we obtain
$$ \forall t \in \mathbb R_+^*, \ \Lambda_{\nu}^*(t) = t \log t - t + 1 .     $$
So we have the following contrast
$$ \mathcal D_3(w,d) = \sum_{i \in s} \pi_i w_i \log (\pi_i w_i) - \pi_i w_i. $$
So we recover the Kullback information where $(\pi_i w_i)_{i \in s}$ is identified with a discrete measures on $s$.
\end{enumerate}
MEM leads to a classical calibration problem where the solution is defined as a minimizer of a convex function subject to linear constraints. The following result gives another expression of the solution for which the computation may be easier in practical cases.

\begin{proposition}\label{prop2} Assume that the optimization problem is feasible, the MEM estimator $\hat w$ is given by:
\begin{equation}
\label{chile3} \forall i \in s, \ \hat w_i = d_i  \Lambda_{\nu_i}' (\hat \lambda^t d_i x_i ) \end{equation} 
where $\hat \lambda$ minimizes over $\mathbb R^k$ $ \lambda \mapsto \sum_{i \in s} \Lambda_{\nu_i} ( \lambda^t d_i x_i)  - \lambda^t t_x$.
\end{proposition}

\noindent We endow $y$ with new weights obtaining the MEM estimator $\hat t_y = N^{-1} \sum_{i \in s} \hat w_i y_i$. We point out that calibration using maximum entropy framework turns into a general convex optimization program, which can be easily solved. Indeed, computing the new weights $w_i, i \in s$, only involves a two step procedure. First, we find the unique $\hat \lambda \in \mathbb R^k$ such that 
$$ N^{-1} \sum_{i \in s} d_i  \Lambda_{\nu_i}' ( \hat \lambda^t d_i x_i) x_i - t_x = 0.$$ 
This is achieved optimizing a scalar convex function. Then, compute the new weights $\hat w_i = d_i  \Lambda_{\nu_i}'(\hat \lambda^t d_i  x_i ) $.

\subsection{Extension to generalized calibration and instrument estimation}\label{sec23}

Proposition \ref{prop2} shows that a calibration estimator is defined using a family of functions $\Lambda_{\nu_i}', i \in s$ satisfying the property that the equation $N^{-1} \sum_{i \in s} d_i  \Lambda_{\nu_i}' (\lambda^t d_i x_i)   x_i \! = \! t_x$ has a unique solution. A natural generalization, known as generalized calibration (GC) (see \cite{MR227827}), consists in replacing the functions $ \lambda \mapsto \Lambda_{\nu_i}' (\lambda^t  d_i x_i)$ by more general functions $f_i: \mathbb R^k \rightarrow \mathbb R, \ i \in s$. Assume that the equation
$$ F(\lambda) =  N^{-1} \sum_{i \in s} d_i f_i(\lambda) x_i = t_x   $$ 
has a unique solution $\hat \lambda$. Assume also that the $f_i$ are continuously differentiable at $0$, and are such that $f_i(0)= 1 $ so that $F(0) = \hat t_x^{HT}$. Then, take as the solution to the generalized calibration procedure, the weights:
$$ \forall i \in s, \ \hat w_i = d_i f_i(\hat \lambda). $$ 
Calibration is of course a particular example of generalized calibration where we set $ f_i: \ \lambda \mapsto \Lambda_{\nu_i}'(\lambda^t d_i x_i)$ to recover a calibration problem seen in Section \ref{sec22}. Even though the method enables a large choice of functions $f_i$, most cases can not be given a probabilistic interpretation. \\
However, an interesting particular choice is given by the functions $\lambda \mapsto 1 + z_i^t \lambda$ for $z_i, i \in s$. This sequence of vectors of $\mathbb R^k$ is called instruments (see \cite{MR227827}). If the matrix $X_n :=N^{-1} \sum_{i \in s} d_i z_i x_i^t$ is invertible, then, the resulting estimator $\hat t_y$, referred to as the instrument estimator obtained with the instruments $z_i$, is given by:
\begin{equation}
\label{inst} \textstyle \hat t_y = \hat t_y^{HT} + (t_x - \hat t_x^{HT})^t X_n^{-1} N^{-1} \sum_{i \in s} d_i z_i y_i.
\end{equation}

\begin{remark}\textbf{: (dimension reduction)} The estimator $\hat t_y$ defined in \eqref{inst}
can be viewed as the instrument estimator obtained with auxiliary variable $\hat B^t x$ and instruments $\hat B^t z_i, i \in s$ with $\hat B = \left[  \sum_{i \in s} d_i z_i x_i^t  \right]^{-1} \sum_{i \in s} d_i y_i z_i $. Hence, in the frame of instrument estimation, the original $k$-dimensional calibration constraint can be replaced by a one-dimensional linearly modified one $ N^{-1} \sum_{i \in U} w_i \hat B^t x_i =  \hat B^t t_x$, without changing the value of the estimator. This enables to reduce the dimension of the problem. Furthermore, it gives an interesting interpretation of the underlying process of calibration. For instance, take the instruments $z_i = x_i, i \in s$. The corresponding variable $B^tx$ is the quadratic projection of $y$ onto the linear space $E_x$, spanned by the components of $x$. In other words, $B^tx$ is a linear approximation of $y$. As a result, the variable $y -B^tx$ has a lower variance than $y$, while its mean over the population $\xi$ is known up to $t_y$. So, the variable $y -B^tx$ can be used to estimate $t_y$ and will provide a more efficient estimator. Since $B$ is unknown, we use $\hat B$ to estimate it. Set $\tilde y = y - \hat B^t x$, we have:
$$ \hat t_y - \hat B^tt_x = N^{-1} \sum_{ i \in s} d_i \tilde y_i. $$
The calibrated estimator $\hat t_y$ appears as the Horvitz-Thompson estimator (up to a known additive constant, here $\hat B^tt_x$) of a variable $\tilde y$ with a lower variance than $y$. This points out that calibration relies on linear regression, since an estimator of $t_y$ is computed by first constructing a linear projection $\hat B^tx$ of $y$ on a subspace $E_x$. Reducing the dimension of the problem is made by choosing the proper real-valued auxiliary variable, and therefore, the proper one-dimensional linear subspace on which $y$ is projected. 
\end{remark}

Note also that the accuracy of the estimator heavily relies on the linear correlation between $y$ and the auxiliary variable. It appears that the accuracy could be improved for some non-linear transformation, say $u(x)$, of the original auxiliary variable $x$, provided that $y$ is more correlated with $u(x)$ than $x$. This is discussed in Section \ref{sec3}.\\

Instrument estimators play a crucial role when studying the asymptotic properties of generalized calibration estimation. A classical asymptotic framework in calibration is to consider that $n$ and $N$ simultaneously go to infinity while the Horvitz-Thompson estimators $\hat t_x^{HT}$ and $\hat t_y^{HT}$ converge at a rate of convergence of $\sqrt{n}$, as described in \cite{MR1173804} and \cite{MR2024768} for instance. This will be our framework here. That is
$$ \Vert \hat t_x^{HT} - t_x \Vert = O_{\mathbb P}(n^{-1/2}) \ \text{ and } \ (\hat t_y^{HT} - t_y) = O_{\mathbb P}(n^{-1/2}). $$
In this framework, all GC estimators are $\sqrt n$-consistent, as seen in \cite{MR1173804}. 

\begin{definition} We say that two GC estimators $\hat t_y$ and $\tilde t_y$ are asymptotically equivalent if $(\hat t_y - \tilde t_y) = o_{\mathbb P}(n^{-1/2})$. 
\end{definition}

\begin{proposition}\label{marcel} Let $\hat t_y$ and $\tilde t_y$ be the GC estimators obtained respectively with the functions $f_i, i \in s$ and $g_i, i \in s$. If for all $i \in s$, $\nabla f_i(0) = \nabla g_i(0) = z_i$, and if the matrix $X_n :=N^{-1} \sum_{i \in s} d_i z_i x_i^t$ converges toward an invertible matrix $X$, then $\hat t_y$ and $\tilde t_y$ are asymptotically equivalent. In particular, two MEM estimators are asymptotically equivalent as soon as their prior distributions have the same respective variances.
\end{proposition}
This proposition is a consequence of Result $3$ in \cite{MR1173804}. It states that for all GC estimator, there exists an instrument estimator having the same asymptotic behavior, built by taking as instruments the gradient vectors of the criterion functions at $0$: $z_i = \nabla f_i(0), i \in s$. Consequently, a MEM estimator $\hat t_y$ built with prior distributions $\nu_i, i \! \in \! s$ with mean $1$ and respective variances $\pi_i q_i$ for $(q_i)_{ i \in s}$ a given positive sequence, satisfies
$$ \textstyle \hat t_y = \hat t_y^{HT} + (t_x - \hat t_x^{HT})^t \hat B + o_\mathbb P(n^{-1/2}) $$
where $\hat B = \left[\sum_{i \in s} d_i q_i x_i x_i^t \right]^{-1} \sum_{i \in s} d_i q_i x_i y_i$. The negligible term $o_\mathbb P(n^{-1/2})$ is zero for all $n$ for Gaussian priors $\nu_i \sim \mathcal N(1,\pi_i q_i)$, 
which stresses the important role played by the corresponding $\chi^2$ dissimilarity (see Example 1 in Section \ref{sec22bis}). Note also that the Gaussian equivalent $\tilde t_y = \hat t_y^{HT} + (t_x - \hat t_x^{HT})^t \hat B$ is the instrument estimator built with the instruments $z_i = q_i x_i$. This choice of instruments, and in particular the case $q_i = 1$ for all $i \in s$, is often used in practice due to its simplicity and good consistency.

\section{Efficiency of calibration estimator with MEM method}\label{sec3}
By using the auxiliary variable $x$ in the calibration constraint, we implicitly assume that $x$ and $y$ are linearly related. However, other relationships may prevail between the variables and it may be more accurate to consider some other auxiliary variable $u(x)$. Here, we discuss optimal choices of function $u: \mathcal X \rightarrow \mathbb R^d$ to use in the calibration constraint. To do so, we first define a notion of asymptotic efficiency in our model with fixed auxiliary variable $u(x)$. Then, we  study the influence of the choice of the constraint function $u$ and find the optimal choice leading to the most efficient estimator. Finally, we propose a method based on the approximate maximum entropy on the mean which enables to compute an asymptotically optimal estimate of $t_y$, taking into consideration both the choice of the constraint function $u$ and the instruments $z_i$.  

\subsection{Asymptotic efficiency}\label{sec31}
In order to choose between calibration estimators, we now define a notion of asymptotic efficiency for a given calibration constraint. Although a GC estimator is entirely determined by a family $f_i, i \in s$ of functions, only the values $z_i = \nabla f_i(0), i \in s$ matter to study the asymptotic behavior of the estimator, up to a negligible term of order $o_\mathbb P(n^{-1/2})$. Let $u: \mathcal X \rightarrow \mathbb R^d$ be a given function, and consider:
$$  t_u = N^{-1} \sum_{i \in U} u(x_i), \ \hat t_{u \pi} = N^{-1} \sum_{i \in s} d_i u(x_i). $$
We make the following assumptions.
\begin{description}
\item \textbf{A1}: $\!\!$ $\xi := \left\lbrace (x_i,y_i), i \in U \right\rbrace $ are independent realizations of  $(X,Y)$, with $\mathbb{E}(Y \vert X) \neq \mathbb{E}(Y )$ and $\mathbb E ( \vert Y^3 \vert) < \infty$. Note respectively $P_X$ and $P_{XY}$ the distributions of $X$ and $(X,Y)$.
\item \textbf{A2}: The sampling design $p(.)$ does not depend on $\xi$.
\item \textbf{A3}: $ n $ and $N/n$ tend to infinity. This will be denoted by $(n, N/n) \rightarrow \infty$. 
\end{description} 
Furthermore, $u$ is assumed to be measurable and such that $\mathbb E ( \Vert u(X)^3 \Vert) < \infty$. Given the constraint function $u$ and instruments $z_i, i \in s$, we note $\hat t_y(u)$ the resulting instrument estimator, the dependency in $z_i$ is dropped for ease of notation. We now study the asymptotic behavior of $\hat t_y(u)$ with respect to the instruments $z_i, i \in s$. Here, the weights $\hat w$ are adapted to the new calibration constraint $ N^{-1} \sum_{i \in U}  \hat w_i u(x_i) = t_u$, yielding
$$ \hat t_y(u) = N^{-1} \sum_{i \in U} \hat w_i y_i = \hat t_y^{HT} + (t_u - \hat t_{u \pi})^t \hat B_u, $$
where $\hat B_u = \left[ \sum_{i \in s} d_i  z_i u(x_i)^t  \right]^{-1}  \sum_{i \in s} d_i y_i z_i $ is assumed to be well defined and to converge in probability towards a constant vector $B_u$ as $(n, N/n) \rightarrow \infty$.\\

In order to define a criterion of efficiency, we first need to construct an asymptotic variance lower bound for instrument estimators. Note $ \mathbb{E}_\xi ( t_y - \hat t_{y}(u) )^2  $ the quadratic risk of $\hat t_y(u)$ under $p$, the population $\xi$ being fixed, we aim to determine a lower bound for the limit of $ n \mathbb{E}_\xi ( t_y - \hat t_{y}(u) )^2$ as $(n, N/n) \rightarrow \infty$ (provided that the limit exists). The value of the limit of course heavily relies on the asymptotic behavior of the sampling design. Without some control on the Horvitz-Thompson weights $\pi_i$, we can not derive consistency properties for instrument estimators. Note $\pi_{ij} = \sum_{s: \ i,j \in s} p(s)$ the joint inclusion probability of $i$ and $j$ and let $\Delta_{ij} = \pi_{ij} d_i d_j - 1    $, we make the following technical assumptions.
\begin{description}
\item \textbf{A4}: $  \textstyle \sum_{i \in U} \Delta_{ii}^2 = o(N^4 n^{-2})$, $ \textstyle \sum_{i \in U} \sum_{j \neq i } \Delta_{ij}^2 = o( N^3 n^{-2})$.
\item \textbf{A5}: $ \! \! \displaystyle \lim_{n \rightarrow \infty \atop N/n \rightarrow \infty} \textstyle  n N^{-2}  \sum_{i \in U} \Delta_{ii} = -  \displaystyle \lim_{{n \rightarrow \infty \atop N/n \rightarrow \infty} } \textstyle n N^{-2} \sum_{i \in U} \sum_{j \neq i } \Delta_{ij} = 1$.
\end{description}
Assumption 4 is sufficient to ensure that the HT estimator of some variable $ a(x_i,y_i), i \in U $ is $\sqrt{n}$-consistent provided that $\mathbb E( a(X,Y)^2) < \infty$. Furthermore, Assumption 5 ensures the existence of its asymptotic variance. Note that these assumptions do not take into consideration the population $\xi$, so that it makes them easy to check in practical cases. For example, the assumptions are fulfilled for the uniform sampling design, that is when $p$ is such that every sample $s \subset U$ has the same probability of being observed. In that case, the Horvitz-Thompson weights are $ \pi_i = n / N$ and $ \pi_{ij} = n(n-1)/N(N-1),  \forall i \! \neq \! j$, yielding $ \Delta_{ii}= N/n -1$ and $ \Delta_{ij} = - (N - n) / n(N - 1)$. We can now state our first result.\\

\noindent Lemma 1: \textit{Suppose that Assumptions 1 to 4 hold. Then, }
 $$  n \mathbb{E}_\xi ( t_y - \hat t_{y}(u) )^2 \geq \text{var}\left( Y - B_u^t u(X) \right) + o_\mathbb P(1),$$
\textit{with equality if, and only if, Assumption 5 also holds.}\\

We point out that an asymptotic lower bound for the variance can be defined for instrument estimators as soon as Assumptions 1 to 4 hold. The lower bound (denoted by $V^*(u)$) is the minimum of $\text{var}(Y-B^t u(X))$ for $B$ ranging over $\mathbb R^d$. It can be computed explicitly if the matrix var$(u(X))$ is invertible: 
$$ V^*(u) = \text{var}\left( Y - \text{cov}(Y,u(X)) ^t \left[ \text{var}(u(X)) \right]^{-1} u(X) \right). $$
We say that an estimator $\hat t_y(u)$ is asymptotically efficient if its asymptotic variance is $V^*(u)$. Note that this lower bound can not be reached if Assumption 5 is not true. We now come to our second result.\\

\noindent Lemma 2: \textit{Suppose that Assumptions 1 to 5 hold. If} var$(u(X))$ \textit{is invertible, $\hat t_y(u)$ built with instrument $z_i, i \in s$ is asymptotically efficient if, and only if, }
\begin{eqnarray}\label{eq6}  \lim_{(n, N/n) \rightarrow + \infty} \textstyle  \left[ \sum_{i \in s} d_i  z_i u(x_i)^t \right]^{-1} \sum_{i \in s} d_i y_i z_i =   \left[ \text{var}(u(X)) \right]^{-1} \text{cov}(Y,u(X)). 
\end{eqnarray} 

In an asymptotic concern and when the calibration function $u$ is fixed, finding the best instruments $z_i, i \! \in \!  s$ in order to estimate $t_y$ becomes a simple optimization problem which depends only on the limit $B_u$ of $\hat B_u = \left[ \sum_{i \in s} d_i  z_i u(x_i)^t \right]^{-1}  \sum_{i \in s} d_i y_i z_i $. Asymptotic efficiency is obtained by choosing instruments minimizing the asymptotic variance. Hence, calculating $B_u$ provides an efficient and easy way to prove the asymptotic efficiency of an instrument estimator. Moreover, this criterion of asymptotic efficiency can be extended to the set of all generalized calibration estimators, as a consequence of Proposition \ref{marcel}. A GC estimator defined by the functions $f_i, i \! \in \! s$ is asymptotically efficient if and only if the vectors $z_i = \nabla f_i(0), i \! \in \! s$ satisfy \eqref{eq6}.\\

\noindent \textit{Proof of Lemmas 1 and 2:} First compute the quadratic risk of $\hat t_y(u)$. Due to its non linearity it is a difficult task. We rather consider its linear asymptotic expansion $\hat t_{y, \lin}(u) := \hat t_y^{HT} + (t_u - \hat t_{u \pi})B_u$ where we recall that $B_u$ is the limit (in probability) of $\hat B_u$. Note that the random effect is due to the sampling design $p$, the population $\xi$ is fixed. We obtain after calculation the following expression for the quadratic risk
$$ \mathbb{E}_\xi ( t_y - \hat t_{y, \lin}(u) )^2 = N^{-2}  \sum_{i,j \in U} \Delta_{ij} \ (y_i - B_u^t u(x_i))(y_j - B_u^t u(x_j)) . $$
Then, the results follow directly from Lemma \ref{jose}, given in the Appendix. \qed \\
\noindent We now see some examples of well-used estimators.\\

\noindent \textbf{Asymptotic variance of some GC estimators}
\begin{enumerate}
\item Optimal instruments.\\
Assume for sake of simplicity that $u$ is real-valued. We denote by $B_{u}^{min}$ the value of $B_u$ achieving the minimal value of the quadratic risk:
$$  B_{u}^{min} = \textstyle \dfrac{ \sum_{i,j \in U} \Delta_{ij} \ u(x_j) y_i}{\sum_{i,j \in U} \Delta_{ij} \ u(x_i) u(x_j)} = \dfrac{\sum_{i \in U} y_i ( \sum_{j \in U} \Delta_{ij} \ u(x_j) ) }{\sum_{i \in U} u(x_i) ( \sum_{j \in U} \Delta_{ij} \ u(x_j) )}.  $$
The corresponding instruments are $ z_i = \sum_{j \in U} \textstyle \Delta_{ij} \ u(x_j), \forall i$. By Lemma \ref{jose}, we see that $B_u^{min}$ converges toward $\text{cov}(Y,u(X))/\text{var}(u(X))$ as $(n,N/n) \rightarrow \infty$, Equation \eqref{eq6} is thus true in that case. If the sampling design is uniform, we obtain after calculation $z_i =\frac{N(N - n)}{n(N-1)} (u(x_i) - t_u)$, and we have:
$$ B_u^{min} = \dfrac{\textstyle \sum_{i \in U} y_i z_i}{ \sum_{i \in U}  z_i u(x_i)} = \dfrac{\text{cov}_e(y,u(x))}{ \text{var}_e(u(x))} $$
where $\text{cov}_e$ and $\text{var}_e$ denote the empirical covariance and variance for the population $\xi$ given by 
$ \text{cov}_e(y,u(x)) = N^{-1} \sum_{i \in U} y_i(u(x_i) - t_u)$ and $\text{var}_e(u(x)) =  \text{cov}_e(u(x),u(x))$. Finally, 
$$  n  \mathbb{E}_\xi ( t_y - \hat t_{y, \lin} )^2 = ( 1 - n N^{-1} ) \ \text{var}_e \left( y - \dfrac{\text{cov}_e(y,u(x))}{\text{var}_e(u(x))} u(x) \right) + o(1).    $$
We have $\lim_{(n,N/n) \rightarrow \infty} n  \mathbb{E}_\xi ( t_y - \hat t_{y, \lin} )^2 = V^*(u)$, as expected. This estimator is thus asymptotically efficient. Although, instruments used for its computation depend on the whole population $(x_i)_{i \in U}$ and therefore, they may be computationally expensive.\\

\item MEM estimators.\\
Take the instruments $z_i = q_i u(x_i), \forall i \in s $ for $(q_i)_{i \in s}$ a positive sequence. As seen in Section \ref{sec22}, these instruments describe the asymptotic behavior of MEM estimators built using prior distributions $\nu_i$ with respective variances $\pi_i q_i$. Even though this choice is often used in practical cases, we see that it does not necessarily lead to an asymptotically efficient estimator $\hat t_y(u)$. Indeed, under regularity conditions on $q_i$ which ensure the convergence of $ \hat B_u$ (basically, the assumptions of Proposition \ref{memop}, which are true for instance if we take $q_i = 1$), we have: 
$$ \hat B_u = \textstyle \left[ \sum_{i \in s} d_i q_i u(x_i)u(x_i)^t \right]^{-1} \sum_{i \in s} d_i q_i y_i u(x_i) \overset{\mathbb P}{\longrightarrow} \left[ \mathbb E(u(X)u(X)^t) \right]^{-1} \mathbb E(Y u(X)) . $$
These instruments satisfy Equation \eqref{eq6} only if 
$$\left[ \mathbb E(u(X)u(X)^t) \right]^{-1} \mathbb E(Y u(X)) = \left[ \text{var}(u(X) \right]^{-1}\text{cov}(Y,u(X)).$$
This is true when $u(.) = \mathbb E (Y \vert X=.)$ or for any $u$ such that $\mathbb E(u(X))=0$, MEM estimators are thus asymptotically efficient in these cases. When this condition is not fulfilled, an easy method to compute an efficient estimator consists in adding the constant variable $1$ in the calibration constraint. We then consider the MEM estimator $\hat t_y(v)$ where $v = (1,u)^t: \mathcal X \rightarrow \mathbb R^{d+1}$, the calibrated weights now satisfy the constraints 
$$ N^{-1} \sum_{i \in s}  w_i u(x_i) = t_u, \ N^{-1} \sum_{i \in s}  w_i = 1 .$$ 
Here, the matrix $\text{var}(v(X))$ is not invertible although we see after a direct calculation that $V^*(v) = V^*(u)$. So, the auxiliary variable is modified but the asymptotic lower bound is unchanged. Furthermore, the MEM estimator $\hat t_y(v)$ obtained in this way is asymptotically efficient, as it is proved in the following proposition.
\end{enumerate}

\begin{proposition}\label{memop} Suppose that Assumptions 1 to 5 hold. Let $(\nu_i)_{i \in s}$ be a family of probability measures with mean $1$ and respective variance $q_i \pi_i$ with $(q_i)_{i \in s}$ a given positive sequence. Assume that there exists $ \kappa = \kappa(n,N) \in \mathbb R$ such that $\kappa \sum_{i \in s} q_i d_i$ is bounded away from zero and $\kappa^2 \sum_{i \in s} (q_i d_i)^2 \rightarrow 0$ as $(n,N/n) \rightarrow + \infty$. Let $v =(1,v_1,...,v_d): \mathcal X \rightarrow \mathbb R^{d+1}$ be a map, where $1,v_1,...,v_d$ are linearly independent. Then, the MEM estimator built with prior distribution $\nu = \otimes_{i \in s} \nu_i$ and calibration constraint $N^{-1} \sum_{i \in s} w_i v(x_i) = t_v$ is asymptotically efficient.
\end{proposition}

\subsection{Approximate Maximum Entropy on the Mean}\label{sec32}

We now turn on the optimal choice of the auxiliary variable $u(x)$ defining the calibration constraint. For a given constraint function $u$, we implicitly take asymptotically optimal instruments $z_i, i \! \in \! s$, that is, instruments such that the resulting estimator $\hat t_y(u)$ has asymptotic variance $V^*(u)$. Hence, minimizing the asymptotic variance of GC estimators with respect to $u$ and $(z_i)_{i \in s}$ reduces to minimizing $V^*(u)$ with respect to $u$.\\
In an asymptotic framework, $u$ can be taken with values in $\mathbb R$ without loss of generality, as discussed in Section \ref{sec23}. So, for a real valued constraint function $u$, $V^*(u)$ is defined as:
$$   V^*(u) = \underset{B \in \mathbb R}{\text{inf}} \ \text{var}(Y - B u(X)) = \text{var}\left( Y - \dfrac{\text{cov}(Y,u(X))}{\text{var}(u(X))} u(X) \right) .    $$
A function $v$ for which $V^*(v)$ is minimal over the set $\sigma_X$ of all real $X$-measurable functions has the form $v(.) = \alpha \mathbb E(Y \vert X = .) + \beta $ for some $(\alpha,\beta) \in \mathbb R^* \times \mathbb R$. Hence, the conditional expectation $\Phi(x) = \mathbb E(Y \vert X=x)$ (or any bijective affine transformation of it) turns out to be the best choice for the auxiliary variable in term of asymptotic efficiency. In that case, the asymptotic lower bound is given by:
$$ V^* = \underset{u \in \ \sigma_X}{\text{min}} V^*(u) = \mathbb E (Y - \mathbb E (Y \vert X))^2.$$
 
For practical applications, this result is useless since the conditional expectation $\Phi$ depends on the unknown distribution of $(X,Y)$. If $\Phi$ were known, the problem of estimating $t_y$ would be easier since the observed value $t_\Phi = N^{-1} \sum_{i \in U}  \Phi(x_i)$ is a $\sqrt N$-consistent estimator of $t_y$ and is therefore much more efficient than any calibrated estimator. When the conditional expectation $\Phi$ is unknown, a natural solution is to replace $\Phi$ by an estimate $\Phi_m$, and then plug it into the calibration constraint. Under regularity conditions that will be made precise later, we show that this approach enables to compute an asymptotically optimal estimator of $t_y$, in the sense that its asymptotic variance is equal to the lower bound $V^*$ defined above. \\

For all measurable function $u$, we now denote by $\hat t_y(u)$ the MEM estimator of $t_y$ obtained with prior distributions $\nu_i \sim \mathcal N(1,\pi_i)$ and auxiliary variables $u(x)$ and $1$. We recall that $\hat t_y(u)$ is $\sqrt n$-consistent with asymptotic variance $V^*(u)$, as shown in Proposition \ref{memop}. Moreover, we know that the asymptotic variance of MEM estimators $\hat t_y(u)$ is minimal for the unknown value $u = \Phi$. The AMEM procedure consists in replacing $\Phi$ by its approximation $\Phi_m$ in the calibration constraint. The so-obtained AMEM estimator $\hat t_y(\Phi_m)$ is thus quite easily computable but still verifies interesting convergence properties as shown in the next proposition.

\begin{proposition}\label{AMEM} Suppose that Assumptions 1 to 5 hold. Let $(\Phi_m)_{m \in \mathbb N}$ be a sequence of functions independent with $\xi$ and such that
$$ \mathbb E (\Phi(X) - \Phi_m(X))^2 = O(\varphi_m ^{-1}) \text{ with } \lim_{m \rightarrow \infty} \varphi_m = + \infty.$$
Then, the AMEM estimator $\hat t_y(\Phi_m)$ is asymptotically optimal among all GC estimators in the sense that $n \mathbb E_\xi(t_y - \hat t_y (\Phi_m))^2  $ converges toward $V^*$ as $n, N/n,m \rightarrow \infty$.
\end{proposition}

When applied to this context, approximate maximum entropy on the mean enables to increase the efficiency of calibration estimators when an additional information is available, namely, an external estimate of the conditional expectation function $\Phi$ is observed. Nevertheless, in our model, it is possible to obtain similar properties under weaker conditions. 

\begin{corollary}\label{AMEM2} Suppose that Assumptions 1 to 5 hold. Let $(\Phi_m)_{m \in \mathbb N}$ be a sequence of functions satisfying
$$ i) \ n \mathbb E_\xi(\hat t_{\Phi \pi} \!  - t_{\Phi} - (\hat t_{\Phi_m \pi}  - t_{\Phi_m}))^2 \stackrel{\mathbb P} { \underset{(n,N/n,m) \rightarrow \infty}{\longrightarrow} } 0 \ \text{ and } \ ii) \ 
 \hat B_{\Phi_m} \stackrel{\mathbb P} { \underset{(n,N/n,m) \rightarrow \infty}{\longrightarrow} } 1.  $$ 
Then, the estimator $\hat t_y (\Phi_m)$ is asymptotically efficient.
\end{corollary}

This corollary does not rule out that the functions $\Phi_m$ are estimated using the data, which was not the case in Proposition \ref{AMEM}. Hence, it becomes possible to compute an asymptotically efficient estimator of $t_y$ without external estimator $\Phi_m$ of $\Phi$. A data driven estimator $\Phi_n$ provides as well an asymptotically efficient estimator of $t_y$, as soon as the two conditions of Corollary \ref{AMEM2} are fulfilled.\\

Now consider an example of AMEM estimator for which the computation is particularly simple, and that provides interesting interpretations. We assume for simplicity that the sampling design is uniform, here $\hat t_y^{HT}$ is simply equal to $N^{-1} \sum_{i \in s} y_i$. Let $(\phi^1,\phi^2,...)$ be a linearly independent total family of $\mathbb L^2(P_X)$. That is, for all measurable function $f: \mathbb R^k \rightarrow \mathbb R$ such that $E(f(X)^2) < \infty$, there exists a unique sequence $(\alpha_n)_{n \in \mathbb N}$ such that
$$ f(X) = \mathbb E(f(X)) + \sum_{i \in \mathbb N} \alpha_i [\phi^i (X) - \mathbb E(\phi^i (X))]. $$
For all $m$, the projection $ \Phi_m$ of $\Phi$ on  $\text{vect}\left\lbrace 1, \phi^1,...,\phi^m \right\rbrace $ is given by 
$$ \Phi_m(.) = \mathbb E(Y) + \text{cov}(Y, \phi_m(X))^t  \left[  \text{var}(\phi_m(X)) \right]^{-1} \left[ \phi_m(.) - \mathbb E (\phi_m(X)) \right] $$
where $\phi_m = (\phi^1,...,\phi^m)^t$. When $n$ is large enough in comparison to $m$, we can define a natural projection estimator $\Phi_{m,n}$ of $\Phi$ as
$$ \Phi_{m,n}(.) =  \hat t_y^{HT} + \hat B_{\phi_m}^t \left[  \phi_m(.) - \hat t_{\phi_m \pi}\right]    $$
where $ \!  \hat B_{\phi_m} \! \! = \!  \! \textstyle \left[ \sum_{i \in s} y_i (\phi_m(x_i) \!  - \hat t_{\phi_m \pi}) \right]^t \! \left[ \sum_{i \in s}  \phi_m(x_i) (\phi_m(x_i) \!  - \hat t_{\phi_m \pi})^t \right] ^{-1} \! \! \!$. \\
We now consider the AMEM estimator $\hat t(\Phi_{m,n})$:
$$\hat t_y(\Phi_{m,n}) = \hat t_y^{HT} + \dfrac{\sum_{i \in s} y_i (\Phi_{m,n}(x_i) - \hat t_{\Phi_{m,n} \pi})}{\sum_{i \in s} \Phi_{m,n}(x_i) (\Phi_{m,n}(x_i) - \hat t_{\Phi_{m,n} \pi})} (t_{\Phi_{m,n}} - \hat t_{\Phi_{m,n} \pi}) $$
which, after simplification, gives 
$$ \hat t_y(\Phi_{m,n}) = \hat t_y^{HT} +  \hat B_{\phi_m}^t (t_{\phi_{m}} - \hat t_{\phi_{m} \pi}) = t_{\Phi_{m,n}}.$$
The objective is to find a path $(m(n),n)_{n \in \mathbb N}$ for which the estimator $\Phi_n := \Phi_{m(n),n}$ satisfies the conditions of Corollary \ref{AMEM2}. We know that, for all $ m$:
\begin{eqnarray*} & & n \mathbb E_\xi(\hat t_{\Phi \pi}  - t_{\Phi} - (\hat t_{\Phi_{m,n} \pi}  - t_{\Phi_{m,n}}))^2 \\
& = & n \mathbb E_\xi(\hat t_{\Phi \pi} - t_{\Phi} + (t_{\phi_m}  - \hat t_{\phi_m \pi} )^t \hat B_{\phi_m})^2 \\
& = & \textstyle n N^{-2} \sum_{i,j \in U} \textstyle \Delta_{ij} (\Phi(x_i) - B_{\phi_m}^t \phi_m(x_i)) (\Phi(x_j) -  B_{\phi_m}^t \phi_m(x_j)) + o_{\mathbb P}(1)
\end{eqnarray*}
where $B_{\phi_m} = \lim_{(n,N/n) \rightarrow \infty} \hat B_{\phi_m} = \text{cov}(Y, \phi_m(X))^t  \left[  \text{var}(\phi_m(X)) \right]^{-1}  $. By Lemma \ref{jose}, we get: 
$$ \forall m, \ n \mathbb E_\xi(\hat t_{\Phi \pi}  - t_{\Phi} - (\hat t_{\Phi_{m,n} \pi}  - t_{\Phi_{m,n}}))^2 \stackrel{\mathbb P} { \underset{(n,N/n) \rightarrow \infty}{\longrightarrow} } \text{var} (\Phi(X) - \Phi_m(X)). $$
Since the convergence is true for all $m$, we can extract a sequence of integers $(m(n))_{n \in \mathbb N}$ such that $\Phi_n := \Phi_{m(n),n}$ undergoes the first condition of Corollary \ref{AMEM2}:
$$ n \mathbb E_\xi(\hat t_{\Phi \pi}  - t_{\Phi} - (\hat t_{\Phi_n \pi}  - t_{\Phi_n}))^2 \stackrel{\mathbb P} { \underset{(n,N/n) \rightarrow \infty}{\longrightarrow} } 0. $$
The second condition of Corollary \ref{AMEM2} is verified for such a sequence $(\Phi_n)_{n \in \mathbb N}$ since for all $n$, $\hat B_{\Phi_n} = 1$. So finally we conclude that the AMEM estimator $\hat t(\Phi_n)$ is asymptotically optimal. 
\begin{remark}\textbf{:} The AMEM estimator is obtained by plugging an estimator $\Phi_n$ of $\Phi$ in the calibration constraint. Note that $ \hat t_y(\Phi_{n})$ is the MEM estimator we obtain with constraint function $(1,\phi_{m(n)}^t)^t$. Indeed, $ \hat t_y(\Phi_{n}) = \hat t_y^{HT} +  \hat B_{\phi_{m(n)}}^t (t_{\phi_{m(n)}} - \hat t_{\phi_{m(n)} \pi})$. This is a consequence of the dimension reduction property relative to instrument estimators discussed in Section \ref{sec23}, $\Phi_{n}$ is an affine approximation of $y$ by the components of $\phi_{m(n)}(x) $. By increasing properly the number of constraints, the projection will converge toward the conditional expectation $\Phi(x)$ yielding an efficient estimator of $t_y$. \\
We can also rewrite the estimator as $ \hat t_y(\Phi_{n}) =  t_{\Phi_{n}}$. In these settings, we can interpret the AMEM procedure as building an estimator of $t_\Phi$ instead of estimating $t_y$. Because $\Phi(x)$ is not a function of $y$, it can be estimated by the empirical mean over the whole population $U$. An estimator of $t_\Phi$ will asymptotically yield an estimate of $t_y$ as a consequence of the relation $\mathbb E(\mathbb E( Y \vert X)) = \mathbb E(Y)$.
\end{remark}

\section{Numerical simulations} \label{sec4}
We shall now give some numerical applications of our results. We made a simulation of a population $U$ of size $N=100000$, where $X$ is a uniform variable on the interval $[1;2]$, and we take $Y = \exp(X) + \varepsilon$ with $\varepsilon \sim \mathcal N(0, \sigma^2)$ an independent noise. So, the conditional expectation $\Phi$ mentioned in the last section is simply the function $\exp(.)$. The sampling design is uniform and the sample $s$ is taken of size $121$. We consider six instruments estimators, $\hat t_1$ to $\hat t_6$, of which we make $50$ realizations observed from $50$ different samples drawn from the fixed population $U$, and we give for $i=1,...,6$ an estimator $V_i$ of the variance calculated from the $50$ observations. The first estimator considered $\hat t_1$ is the Horvitz-Thompson estimator, and the last one $\hat t_6$ is the AMEM estimator taken as example in Section \ref{sec32}, where we took the family $\left\lbrace X^i: i \in \mathbb N \right\rbrace $ for the base of $\mathbb L^2(P_X)$, and we set the number $m$ of constraint functions to $m=6$. The construction of the estimators are detailed in the following table. The results are given for two different values of $\sigma^2$, namely $\sigma^2 = 1$ and $\sigma^2 = 0.1$. \\
 
\noindent $1$. $\varepsilon \sim \mathcal N(0,1)$:\\

\noindent \begin{tabular}{|c|c|c|c|}
\hline      & auxiliary variable & instrument & estimated variance \\ 
\hline $\hat t_1$ (H-T estimator) &    none      &    none        & $V1=2.07 \times 10^{-2}$ \\ 
\hline $\hat t_2$ & $x$  & $(x_i)_{i \in s}$  &  $V2=7.8 \times 10^{-3}$ \\ 
\hline $\hat t_3$ & $\mathrm{x} = (1,x)$ & $(\mathrm x_i)_{i \in s}$ & $V3=7.6 \times 10^{-3}$ \\ 
\hline $\hat t_4$ & $\exp(x)$ & $(\exp(x_i))_{i \in s}$ & $V4=7.2 \times 10^{-3}$ \\ 
\hline $\hat t_5$ & $\mathrm{x} = (1,\exp(x))$ & $(\mathrm x_i)_{i \in s}$ & $V5=6.9 \times 10^{-3}$ \\ 
\hline $\hat t_6$ (AMEM estimator) & $ \mathrm{x} = (1,x,x^2,x^3,x^4,x^5,x^6)$ & $ (\mathrm x_i)_{i \in s}$ & $V6=7.2 \times 10^{-3}$ \\ 
\hline 
\end{tabular} 
$$  $$
We observe that the calibrated estimators appear to be better than the Horvitz-Thompson estimator. The choice of the auxiliary variable or the instrument does not seem to have a significant effect on the efficiency.\\

\noindent $2$. $\varepsilon \sim \mathcal N(0,0.1)$:\\

\noindent \begin{tabular}{|c|c|c|c|}
\hline      & auxiliary variable & instrument & estimated variance \\ 
\hline $\hat t_1$ (H-T estimator) &    none      &    none        & $V1=1.93 \times 10^{-2}$ \\ 
\hline $\hat t_2$ & $x$  & $(x_i)_{i \in s}$  &  $V2=3.1 \times 10^{-3}$ \\ 
\hline $\hat t_3$ & $\mathrm{x} = (1,x)$ & $(\mathrm x_i)_{i \in s}$ & $V3=8.7 \times 10^{-4}$ \\ 
\hline $\hat t_4$ & $\exp(x)$ & $(\exp(x_i))_{i \in s}$ & $V4=6.8 \times 10^{-4}$ \\ 
\hline $\hat t_5$ & $\mathrm{x} = (1,\exp(x))$ & $(\mathrm x_i)_{i \in s}$ & $V5=6.7 \times 10^{-4}$ \\ 
\hline $\hat t_6$ (AMEM estimator) & $ \mathrm{x} = (1,x,x^2,x^3,x^4,x^5,x^6)$ & $ (\mathrm x_i)_{i \in s}$ & $V6=7.0 \times 10^{-4}$ \\ 
\hline 
\end{tabular} 
$$  $$
Here, $X$ explains almost entirely $Y$ since the variance of $\varepsilon$ is low ($\sigma^2=0.1$). In that case, the choice of the auxiliary variable and instrument appears to play a more important role. We notice a significant difference between $\hat t_2$ and $\hat t_3$ which points out the importance of the instrument. More specifically, we see that the instrument $(x_i - \hat t_x^{HT})_{i \in s}$ (which is equivalent to adding the constant $1$ as an auxiliary variable) provides a better estimator than $x_i$. Furthermore, also note that using the auxiliary variable $\Phi(x)=\exp(x)$ provides the best estimator in term of minimal variance as we see that $V4$ and $V5$ are the smallest estimated variances. These estimators can be viewed as oracles, since the auxiliary variable used in that case is the optimal choice, but is in general unknown (see Section \ref{sec32}). The difference between $\hat t_4$ and $\hat t_5$ is not significant, as expected, according to the second example of Section \ref{sec31}. Finally, the AMEM estimator has its variance lying between that of the standard calibrated estimator $\hat t_3$ and that of the oracles, which conveys that it is more efficient than $\hat t_3$.

\section{Appendix} \label{sec5}

\subsection{Technical lemma}

\begin{lemma}\label{jose} Let $\mathcal F$ be the set of all functions $f:  (\mathbb R^k \times \mathbb R) \rightarrow \mathbb R$ such that $\mathbb E ( \vert f(X,Y) \vert^3)$ is finite (we set $ f_i = f(x_i,y_i)$ for all $i \in U$). Under Assumptions 1, 2 and 4, 
$$ \forall f \in \mathcal F, \ n N^{-2} \sum_{i,j \in U} \textstyle \Delta_{ij} \ f_i f_j \geq  \emph{var} (f(X,Y)) + o_{\mathbb P}(1) $$ 
as $(n, N/n) \rightarrow \infty$, with equality if and only if Assumption 5 also holds. In that case, the quantity $ n N^{-2} \sum_{i,j \in U} \textstyle \Delta_{ij} \ f_i g_j$ converges in probability toward $\emph{cov}  (f(X,Y), g(X,Y))$ for all $f, g \in \mathcal F$ as $(n, N/n) \rightarrow \infty$.
\end{lemma}

\begin{proof} \textbf{of Lemma \ref{jose}:}\\ 
Assumptions 1, 2 and 4 yield for all $f \in \mathcal F$:
\begin{eqnarray*}
 & &  \textstyle   n N^{-2} \sum_{i,j \in U} \Delta_{ij} \ f_i f_j =   n N^{-2} \sum_{i \in U} \Delta_{ii} \ f_i^2  +   n N^{-2}  \sum_{i \neq j } \Delta_{ij} \  f_i f_j  \\
& = &  \textstyle \left( n N^{-2}  \sum_{i \in U} \Delta_{ii} \right) \mathbb E (f(X,Y) ^2)  +  \left(  n N^{-2}  \sum_{i \neq j } \Delta_{ij} \right)  \mathbb E (f(X,Y) )^2 + o_{\mathbb P}(1) 
\end{eqnarray*}
Let $ \mathcal{P}_n(U)$ denote the set of all subsample $s$ of $U$ with $n$ elements. By Jensen inequality, we get
$$   \textstyle  \sum_{i, j \in U } \Delta_{ij} =  \textstyle  \sum_{s \in \mathcal{P}_n(U)} \left( \textstyle \sum_{i \in s} d_i \right)^2 p(s)  - N^2 \geq  \left[   \sum_{s \in \mathcal{P}_n(U)} \left( \textstyle \sum_{i \in s} d_i \right) p(s) \right] ^2  - N^2  \geq 0 $$
which implies that $  \sum_{i \neq j } \Delta_{ij}  \geq   - \sum_{i \in U} \Delta_{ii} $. Thus: 
$$  \textstyle \ \ n N^{-2}  \sum_{i,j \in U}  \Delta_{ij} \ f_i f_j \geq  \left( n N^{-2} \sum_{i \in U} \Delta_{ii} \right) \text{var} (f(X,Y)) + o_{\mathbb P}(1).$$
Since $\sum_{i \in U} \pi_i = n$, we know that $n N^{-2} \sum_{i \in U} \Delta_{ii} \geq 1 - n N^{-1}$ by convexity of $x \mapsto 1/x$ on $\mathbb R_+^*$. Hence
$$  \textstyle  n N^{-2} \sum_{i,j \in U} \Delta_{ij} \ f_i f_j \geq \text{var} (f(X,Y)) + o_{\mathbb P}(1).$$
as $(n, N/n) \rightarrow \infty$. Furthermore, it is not an equality for all $f \in \mathcal F$ if Assumption 5 is not true. We show the second part of the lemma using the same pattern as in the beginning of the proof applied to $f$ and $g$. In particular, it holds when $f=g$. 
\end{proof}

\subsection{Proofs}

\begin{proof} \textbf{of Theorem \ref{prop1}:}\\ 
For all $w \in \mathbb R^n$, let $f_w: \mathbb R^n \rightarrow \mathbb R_+$ be the unique minimizer of the functional $ f \mapsto K(f \nu, \nu)$ on the set $\mathcal F_w = \left\lbrace f: \int_{\mathbb R^n} (\tau - \pi w)f(\tau) d\nu(\tau) = 0 \right\rbrace  $. We have:
$$  \textstyle  f_w = \underset{f \in \mathcal F_w}{ \text{argmin}} \ \int_{\mathbb R^n} f (\log (f)-1) d \nu.    $$
We calculate the Lagrangian $\mathcal L(\lambda, f)$ associated to the problem:
$$ \textstyle \mathcal L(\lambda, f) = \int_{\mathbb R^n} [ f(\tau) \log (f(\tau)) - f(\tau) ] d \nu(\tau) - \lambda^t \int_{\mathbb R^n} (\tau- \pi w) f(\tau) d\nu(\tau)$$
where $\lambda \in \mathbb R^n$ is the Lagrange multiplier. The first order conditions are:
$$  \textstyle  \forall \tau \in \mathbb R^n, \ \log (f(\tau)) = \lambda^t (\tau- \pi w).$$
Hence, $\forall \tau, \ f_w(\tau) = e^{\lambda_w^t (\tau- \pi w)}$ where $\lambda_w$ verifies:
$$  \textstyle \int_{\mathbb R^n} (\tau- \pi w) e^{\lambda^t (\tau- \pi w) } d\nu(\tau) = 0 \Longleftrightarrow \lambda_w = \underset{\lambda \in \mathbb R^n}{ \text{argmin}} \int_{\mathbb R^n} e^{\lambda^t (\tau- \pi w) } d\nu(\tau)$$
Let $ S = \left\lbrace (w_i)_{i \in s}: \ \textstyle N^{-1} \sum_{i \in s} x_i w_i = t_x \right\rbrace $, we notice that 
\begin{eqnarray*}
\hat w = \textstyle \mathbb E_{\nu^*} (W) & = & \textstyle \underset{w \in S}{ \text{argmin}} \left\lbrace \text{min}_{f \in \mathcal F_w} \  \int_{\mathbb R^n} f (\log (f)-1) d \nu \right\rbrace \\
& = & \textstyle  \underset{w \in S}{ \text{argmin}} \left\{ \int_{\mathbb R^n} f_w (\log (f_w)-1) d \nu \right\} \\
& = & \textstyle  \underset{w \in S}{ \text{argmin}}  \left\{ \lambda_w^t \int_{\mathbb R^n}  (\tau- \pi w) e^{\lambda_w^t (\tau- \pi w) } d\nu(\tau) - \int_{\mathbb R^n} e^{\lambda_w^t (\tau- \pi w)} d\nu(\tau) \right\} \\
& = &  \textstyle  \underset{w \in S}{ \text{argmin}} \left\{ - \text{min}_{\lambda \in \mathbb R^n} \ e^{-\lambda^t \pi w} \int_{\mathbb R^n} e^{\lambda^t \tau}  d\nu(\tau) \right\}.
\end{eqnarray*}
by definition of $\lambda_w$. Recall that $\nu = \otimes_{i \in s} \nu_i$. Since the function $t \mapsto - \log t$ is decreasing, we have that
$$ \textstyle \underset{\lambda \in \mathbb R^n}{\text{min}} \left\{ e^{-\lambda^t \pi w} \int_{\mathbb R^n} e^{\lambda^t \tau}  d\nu(\tau) \right\} = \exp - \underset{\lambda \in \mathbb R^n}{\text{sup}} \left\{ \sum_{i \in s} [ \lambda_i \pi_i w_i - \log \int_{\mathbb R} e^{\lambda_i \tau_i}  d\nu_i(\tau_i)  ]    \right\}    $$
The supremum being taken for $\lambda \in \mathbb R^n$, we see that 
$$ \textstyle \underset{\lambda \in \mathbb R^n}{\text{sup}} \left\{ \sum_{i \in s} [ \lambda_i \pi_i w_i  - \log \int_{\mathbb R} e^{\lambda_i \tau_i}  d\nu_i(\tau_i) ]    \right\} = \sum_{i \in s} \underset{\lambda_i \in \mathbb R}{\text{sup}} \left\{ \lambda_i \pi_i w_i  - \log \int_{\mathbb R} e^{\lambda_i \tau_i}  d\nu_i(\tau_i) \right\} $$
Finally we obtain:
\begin{eqnarray*}
 \hat w = \textstyle  \underset{w \in S}{ \text{argmin}} - \exp \left( - \sum_{i \in s} \Lambda_{\nu_i}^*(\pi_i w_i) \right) =  \textstyle \underset{w \in S}{ \text{argmin}} \ \sum_{i \in s} \Lambda_{\nu_i}^*(\pi_i w_i).
\end{eqnarray*}
\end{proof}

\begin{proof} \textbf{of Proposition \ref{prop2}:}\\ 
It is a classic convex optimization problem. Let $\mathcal L$ be the Lagrangian associated to the problem:
$$  \textstyle \mathcal{L}(\lambda,w) =  \sum_{i \in s} \Lambda_{\nu_i}^*(w_i \pi_i) - \lambda^t \left( \sum_{i \in s} w_i x_i - N t_x \right)  $$
where $\lambda \in \mathbb R^k$ is the Lagrange multiplier. The solutions to the first order conditions satisfy for all $i \in s$,
$$ w_i = d_i ({\Lambda_{\nu_i}^*}')^{-1} ( \lambda^t d_i x_i ),  $$
where we recall that the functions $\Lambda_{\nu_i}^*$ are assumed to be strictly convex, so that $({\Lambda_{\nu_i}^*}')^{-1}$ exists for all $i$, and is equal to $ \Lambda_{\nu_i}'$. Now it suffices to apply the solutions of the first order conditions to the constraint to obtain an expression of the solution $\hat \lambda$: 
\begin{eqnarray*}
 N^{-1}  \textstyle \sum_{i \in s} d_i \Lambda_{\nu_i}'( \hat \lambda^t d_i x_i)  x_i - t_x = 0 \Longleftrightarrow  \hat \lambda = \underset{\lambda \in \mathbb{R}^k}{\text{argmin}} \sum_{i \in s} \Lambda_{\nu_i} ( \lambda^t d_i x_i)  - \lambda^t t_x.
\end{eqnarray*}
The equivalence is justified by the fact that $\Lambda_{\nu_i}$ is strictly convex, and therefore, so is $\lambda \mapsto \sum_{i \in s} \Lambda_{\nu_i} ( \lambda^t  d_i x_i)  - \lambda^t t_x $. For that reason, $\hat \lambda$ is uniquely defined. We finally obtain an expression of the calibrated weights 
$$  \forall i \in s, \ \hat w_i = d_i  \Lambda_{\nu_i}' (\hat \lambda^t d_i x_i ). $$
\end{proof}

\begin{proof} \textbf{of Proposition \ref{marcel}:}\\ 
Let $F: \lambda \mapsto N^{-1} \sum_{i \in s} d_i f_i(\lambda) x_i$, and $G: \lambda \mapsto N^{-1} \sum_{i \in s} d_i g_i(\lambda) x_i$. We call respectively $\hat \lambda$ and $\tilde \lambda$ the solutions to $F(\lambda) = t_x$ and $G(\lambda) = t_x$. We have
$$ F(\hat \lambda) = F(0) + X_n \hat \lambda + o(\Vert \hat \lambda \Vert) $$
and then $ (t_x - \hat t_x^{HT}) = X_n \hat \lambda +  o(\Vert \hat \lambda \Vert)$. By assumption, $X_n$ is invertible for large values of $n$ since it converges towards an invertible matrix $X$. Thus, whenever $\hat t_x^{HT}$ is close enough to $t_x$, there exists $\lambda_0$ in a neighborhood of $0$ such that $F(\lambda_0) = t_x$. By uniqueness of the solution, we conclude that $\lambda_0 = \hat \lambda$. Hence, for large values of $n$, 
$$ \hat \lambda = X_n^{-1} (t_x - \hat t_x^{HT}) + o_\mathbb P (n^{-1/2}). $$
A similar reasoning for $\tilde \lambda$ yields $ \Vert \tilde \lambda - \hat \lambda \Vert =  o_\mathbb P (n^{-1/2})$. Thus, $ \hat \lambda $ and $\tilde \lambda$ converge toward $0$ and by Taylor formula:
$$ f_i(\hat \lambda) = 1 + z_i^t \hat \lambda + o_\mathbb P (n^{-1/2}) = 1 + z_i^t \tilde \lambda + o_\mathbb P (n^{-1/2}) = g_i (\tilde \lambda) +  o_\mathbb P (n^{-1/2}).  $$
It follows that $\hat t_y$ and $\tilde t_y$ are asymptotically equivalent.\\ 
We know that MEM estimation reduces to taking $f_i(.) = \Lambda_{\nu_i}'(d_i x_i^t .)$ in a GC procedure. Hence, in that case, $\nabla f_i(0) = d_i \Lambda_{\nu_i}''(0) x_i$. Since the variance of a probability measure $\nu_i$ is given by $\Lambda_{\nu_i}''(0)$, two MEM estimators with prior distributions having the same respective variances are asymptotically equivalent. Furthermore, a Gaussian prior $\nu_i \sim \mathcal N(1,q_i \pi_i) $ has a log-Laplace transform $\Lambda_{\nu_i}: t \mapsto \pi_i q_i t^2 /2 + t$ which corresponds to $f_i(\lambda) = \Lambda_{\nu_i}'(d_i x_i^t \lambda) = 1 + q_i x_i^t \lambda$. The resulting MEM estimator is thus the instrument estimator obtained with instruments $z_i = q_i x_i, i \in s$.
\end{proof}

\begin{proof} \textbf{of Proposition \ref{memop}:}\\ 
We set $u = (v_1,...,v_{d})$, the matrix $\text{var} (u(X))$ is invertible. By assumption on $(q_i)_{i \in s}$, we have 
$$ \textstyle \kappa \sum_{i \in s} d_i q_i y_i v(x_i) = (\kappa \sum_{i \in s} d_i q_i) \mathbb E(Y v(X)) + \kappa \ o_\mathbb P (1)$$
and
$$ \textstyle \kappa \sum_{i \in s} d_i q_i v(x_i) v(x_i)^t = (\kappa \sum_{i \in s} d_i q_i) \mathbb E(v(X) v(X)^t) + \kappa \ o_\mathbb P (1).$$
Since $(\kappa \sum_{i \in s} d_i q_i)$ is bounded away from zero, it follows that 
$$\hat B_v = \textstyle \left[ \textstyle \sum_{i \in s} d_i q_i v(x_i) v(x_i)^t \right]^{-1} \sum_{i \in s} d_i q_i y_i v(x_i) \overset{\mathbb P}{\longrightarrow} \left[ \mathbb E(v(X)v(X)^t) \right]^{-1}\mathbb E (Y v(X)) = B_v.$$
By simple algebra, we show the functional equality $B_v^tv(.) = B_u ^t u(.) + K$,
where $K $
is constant, and therefore does not modify the value of the variance. More precisely, the asymptotic variance of $\hat t_y(v)$ is
$$ \text{var}(Y- \text{cov}(Y,u(X))^t \left[ \text{var}(u(X) \right]^{-1} u(X) + K) = V^*(u), $$
which proves that the MEM estimator $\hat t_{y}(v)$ is asymptotically efficient.
\end{proof}

\begin{proof} \textbf{of Proposition \ref{AMEM}:}\\ 
We decompose the AMEM estimator as follow
$$ \hat t_y(\Phi_m) = \hat t_y^{HT} + (t_{\Phi} - \hat t_{\Phi \pi}) + (\hat t_{\Phi \pi} \!  - t_{\Phi} - (\hat t_{\Phi_m \pi}  - t_{\Phi_m})) + (\hat B_{\Phi_m} \! \! - 1) (t_{\Phi_m} \! \! - \hat t_{\Phi_m \pi}). $$
We have by assumption
\begin{eqnarray*}
n \mathbb E_\xi(\hat t_{\Phi \pi} \!  - t_{\Phi} - (\hat t_{\Phi_m \pi}  - t_{\Phi_m}))^2 = O_{\mathbb P}(\varphi_m^{-1}) \ \text{ and } \ 
(\hat B_{\Phi_m} \! - 1) = O_{\mathbb P}(\varphi_m^{-1/2}) 
\end{eqnarray*}
as $n, N/n \rightarrow \infty$ and uniformly for all $m$ (see the proof of Lemma 1 in \cite{l2}). Hence, the terms $(\hat t_{\Phi \pi} \!  - t_{\Phi} - (\hat t_{\Phi_m \pi}  - t_{\Phi_m}))$ and $(\hat B_{\Phi_m} \! \! - 1) (t_{\Phi_m} \! \! - \hat t_{\Phi_m \pi})$ are asymptotically negligible in comparison to $(t_\Phi - \hat t_{\Phi \pi})$ as $n, N/n, m \rightarrow \infty$. We conclude using Result 2 and Lemma \ref{jose}. 
\end{proof}

\begin{proof} \textbf{of Corollary \ref{AMEM2}:}\\ 
All conditions are fulfilled so that the proof of Proposition \ref{AMEM} remains valid in that case. 
\end{proof}

\bibliographystyle{plain}
\bibliography{survey}

\end{document}